\documentclass[onecolumn, draft, 12pt]{IEEEtran}
\usepackage{amsmath,cases}
\usepackage[mathscr]{euscript}
\usepackage{amsthm}
\usepackage{amssymb}
\usepackage{cite}
\usepackage{amsfonts}
\usepackage{enumerate}
\usepackage[final]{graphicx}

\newtheorem{theorem}{Theorem}

\newtheorem{corollary}{Corollary}

\newtheorem{definition}{Definition}
\theoremstyle{remark}

\makeatletter

\newcommand{\Rmnum}[1]{\expandafter\@slowromancap\romannumeral #1@}
\makeatother

\begin{document}

\title{Stochastic Ordering of Interference in Large-scale Wireless Networks}
\author{Junghoon Lee and Cihan Tepedelenlio\u{g}lu, \emph{Member, IEEE}
\thanks{The authors are with the School of Electrical, Computer, and Energy Engineering, Arizona
State University, Tempe, AZ 85287, USA. (Email:\{junghoon.lee,cihan\}@asu.edu).} }
\maketitle

\begin{abstract}
Stochastic orders are binary relations defined on probability distributions which capture intuitive notions like
being larger or being more variable. This paper introduces stochastic ordering of interference distributions in large-scale networks modeled as point processes. Interference is a major performance-limiting factor in most wireless networks, thus it is important to characterize its statistics. Since closed-form results for the distribution of interference for such networks are only available in limited cases, it is of interest to compare network interference using stochastic orders, for two different point processes with different fading or path-loss scenarios between the interferers and the receiver. In this paper, conditions on the fading distribution and path-loss model are given to establish stochastic ordering between interferences. Moreover, Laplace functional ordering is defined between point processes and applied for comparing interference. Monte-Carlo simulations are used to supplement our analytical results.

\end{abstract}


\section{Introduction}
Since interference is the main performance-limiting factor in most wireless networks, it is crucial to characterize its statistics. The interference mainly depends on the fading channel (interfering power distribution), the path-loss model (signal attenuation with distance), and network geometry (spatial distribution of concurrently transmitting nodes). The spatial location of the interferers can be modeled either deterministically or stochastically. Deterministic models include square, triangular, and hexagonal lattices \cite{Silvester1983, Mathar1995}, which are applicable when the location of the nodes in the network is constrained to a regular structure. On the other hand, only a statistical description of the location of the nodes is available in some scenarios. In both cases, the locations of transmitting nodes in the network are seen as the realizations of some point processes \cite{Stoyan1995, Baccelli2009a, Haenggi2008}. For certain classes of node distributions, most notably Poisson point processes, and certain attenuation laws, closed-form results are available for the interference distribution which determine the network performance \cite{Ilow1998, Win2009} (and the references therein). However the interference distribution is not tractable in most other cases.

Successful transmission probability in the presence of interference can be calculated by determining the Laplace transform of interference \cite{Linnartz1992, Baccelli2006, Haenggi2008}. However, closed-form expressions for the Laplace transform of interference are not tractable in many cases. We approach this problem from a stochastic ordering perspective, which is a partial order on random variables \cite{Shaked, Shaked2}. Concepts of stochastic ordering have been applied to scenarios of interest in wireless communications in \cite{Tepedelenlioglu2011}. We will use these concepts to understand the interference in large scale wireless networks. In \cite{Shaked2, Ross1978, Blaszczyszyn2009}, the application of this set of tools in communication networks can be found. In \cite{Shaked2, Ross1978}, the stochastic ordering has been used in studying a class of queueing networks. Directionally convex ordering of different point processes and its integral shot noise fields which are inherent from the point processes has been studied in \cite{Blaszczyszyn2009}. To the best of our knowledge, there is no study of conditions on fading channels and path-loss models for stochastic ordering of network interference in the literature. In this paper, we use stochastic ordering theory to compare network performance under conditions of differing fading on the interference link, and different path-loss models for the establishment of stochastic ordering of interference from different point processes. Using the conditions, we compare performance without having to obtain closed-form results for a wide range of performance metrics. We also compare different point processes which are commonly used in the literature using stochastic orders, and advocate Laplace functional ordering of point processes over directional convex ordering when interferences due to these point processes are compared.

The paper is organized as follows: In Section \ref{sec:Math_Prelim_and_System_Model}, we introduce mathematical preliminaries and present the system model and assumptions. Section \ref{sec:Ordering_of_Perform_Metrics} introduces performance metrics involving stochastic ordering of interference. In Section \ref{sec:Comparison of Fading_Channels} and \ref{sec:Comparison of Path-loss Models}, we derive the conditions on fading channel and path-loss model for stochastic ordering of interference respectively. Section VI defines Laplace functional ordering and compares different point processes. Section \ref{sec:Numerical_Results} presents results from numerical simulations to corroborate our claims. Finally, the paper is summarized in Section \ref{sec:Summary}.

\section{Mathematical Preliminaries and System Model}
\label{sec:Math_Prelim_and_System_Model}
Here we give a brief overview of some terminology and mathematical tools and introduce the system model assumed in this paper.
\subsection{Stochastic Ordering}
First we briefly review some common stochastic orders between random variables, which can be found in \cite{Shaked, Shaked2}.
\subsubsection{Usual Stochastic Ordering}
\label{subsubsec:Usual_Transform_Ordering}
Let $X$ and $Y$ be two random variables (RVs) such that
\begin{equation}
\label{eqn:def_st_ordering}
P\left( X>x \right) \leq P\left( Y>x \right), -\infty < x < \infty.
\end{equation}
Then $X$ is said to be smaller than $Y$ in the \emph{usual stochastic order} (denoted by $X \leq_{\mathrm{st}} Y$). Roughly speaking, \eqref{eqn:def_st_ordering} says that $X$ is less likely than $Y$ to take on large values. To see the interpretation of this in the context of wireless communications, when $X$ and $Y$ are distributions of instantaneous SNRs due to fading, \eqref{eqn:def_st_ordering} is a comparison of outage probabilities. Since $X, Y$ are positive in this case, $x\in\mathbb{R}^+$ is sufficient in \eqref{eqn:def_st_ordering}.

\subsubsection{Laplace Transform Ordering}
\label{subsubsec:Laplace_Transform_Ordering}
Let $X$ and $Y$ be two non-negative random variables such that
\begin{equation}
\label{eqn:def_Lt_ordering}
\mathcal{L}_X(s):=\mathbb{E}[\exp{(-sX)}] \geq \mathbb{E}[\exp{(-sY)}]=\mathcal{L}_Y(s) \text{ for } s>0.
\end{equation}
Then $X$ is said to be smaller than $Y$ in the \emph{Laplace transform} (LT) \emph{order} (denoted by $X \leq_{\mathrm{Lt}} Y$). For example, when $X$ and $Y$ are the instantaneous SNR distributions of a fading channel, \eqref{eqn:def_Lt_ordering} can be interpreted as a comparison of average bit error rates for exponentially decaying instantaneous error rates (as in the case for differential-PSK (DPSK) modulation and Chernoff bounds for other modulations) \cite{Tepedelenlioglu2011}. The LT order $X \leq_{\mathrm{Lt}} Y$ is equivalent to
\begin{equation}
\label{eqn:def_Lt_ordering_conseq1}
\mathbb{E}[l(X)] \geq \mathbb{E}[l(Y)],
\end{equation}
for all \emph{completely monotonic} (c.m.) functions $l(\cdot)$ \cite[pp. 96]{Shaked2}. By definition, the derivatives of a c.m. function $l(x)$ alternate in sign: $(-1)^n \mathrm{d}^n l(x)/\mathrm{d}x^n \geq 0$, for $n=0,1,2,\dots$, and $x \geq 0$. An equivalent definition is that c.m. functions are positive mixtures of decaying exponentials \cite{Shaked2}. A similar result to \eqref{eqn:def_Lt_ordering_conseq1} with a reversal in the inequality states that
\begin{equation}
\label{eqn:def_Lt_ordering_conseq2}
X \leq_{\mathrm{Lt}} Y \Longleftrightarrow \mathbb{E}[l(X)] \leq \mathbb{E}[l(Y)],
\end{equation}
for all $l(\cdot)$ that have a completely monotonic derivative (c.m.d.). Finally, note that $X \leq_{\mathrm{st}} Y \Rightarrow$ \\ $X \leq_{\mathrm{Lt}} Y$. This can be shown by invoking the fact that $X \leq_{\mathrm{st}} Y$ is equivalent to $\mathbb{E}[l(X)] \leq \mathbb{E}[l(Y)]$ whenever $l(\cdot)$ is an increasing function \cite{Shaked2}, and that c.m.d. functions in \eqref{eqn:def_Lt_ordering_conseq2} are increasing.

\subsection{Point Processes}
\label{Point_Processes}
Point processes have been used to model large-scale networks \cite{Karr1991, Stoyan1995, Haenggi2008, Win2009, Haenggi2009b, Tresch2010, Gulati2010, Haenggi2012, Haas2003, Garetto2011}. In this paper, we focus on stationary and isotropic point processes. A point process $\Phi$ is stationary if its distribution is invariant to translations and is isotropic if its distribution is invariant to rotations. In what follows, we introduce some fundamental notions that will be useful.

\subsubsection{Campbell's Theorem \normalfont{\cite{Stoyan1995, Haenggi2008}}}
\label{sec:Campbell_Theorem}
It is often necessary to evaluate the average sum of a function evaluated at the point of the process $\Phi$. Campbell's theorem helps in evaluating such sums. For any non-negative measurable function $u$,
\begin{equation}
\label{eqn:Campbell_Theorem_PP}
\mathbb{E}\left[\sum_{x\in\Phi}u(x)\right]=\int_{\mathbb{R}^d}u(x)\Lambda(\mathrm{d}x) \text{.}
\end{equation}
The intensity measure $\Lambda$ of $\Phi$ in \eqref{eqn:Campbell_Theorem_PP} is a characteristic analogous to the mean of a real-valued random variable and defined as $\Lambda(B)=\mathbb{E}\left[\Phi(B)\right]$ for bounded subsets $B \subset \mathbb{R}^d$. So $\Lambda(B)$ is the mean number of points in $B$. If $\Phi$ is stationary then the intensity measure simplifies as $\Lambda(B)=\lambda \vert B \vert$ for some non-negative constant $\lambda$, which is called the intensity of $\Phi$, where $\vert B \vert$ denotes the $d$ dimensional volume of $B$. For stationary point processes, the right side in \eqref{eqn:Campbell_Theorem_PP} is equal to $\lambda\int_{\mathbb{R}^d}u(x)\mathrm{d}x$. Therefore, any two stationary point processes with same intensity lead to equal average sum of a function (when the mean value exists).

\subsubsection{Laplace Functional of Point Processes \normalfont{\cite{Baccelli2009a}}}
The Laplace functional $L$ of point process $\Phi$ is defined by the following formula
\begin{equation}
\label{eqn:Def_Lf_of_PP}
L_{\Phi}(u) \triangleq \mathbb{E}\left[e^{-\sum_{x\in\Phi}u(x)}\right]=\mathbb{E}\left[e^{-\int_{\mathbb{R}^d}u(x)\Phi(\mathrm{d}x)}\right]
\end{equation}
where $u(\cdot)$ runs over the set $\mathscr{U}$ of all non-negative functions on $\mathbb{R}^d$. The Laplace functional completely characterizes the distribution of the point process. As an important example, the Laplace functional $L$ of stationary Poisson point process $\Phi_{\normalfont{\text{PPP}}}$ is
\begin{equation}
\label{eqn:Def_Lf_of_PPP}
L_{\Phi_{\normalfont{\text{PPP}}}}(u) = \exp\left\{-\lambda\int_{\mathbb{R}^d}\left[1-\exp(-u(x))\right]\mathrm{d} x\right\}
\end{equation}
where $\lambda$ is the intensity.

\subsection{System Model}
\label{subsec:System_Model}
As shown in Fig. \ref{fig:system_model}, we assume a transmit/receive pair communicating over a wireless channel. The receiver is being interfered by interference sources distributed as stationary and isotropic point process. The point process describes all interfering nodes \cite{Gulati2010}. Both the transmitter and receiver are fixed and are not considered a part of the point process. The accumulated interference to the receiver at the origin is of interest to quantify and it is given by
\begin{equation}
\label{eqn:Interference_model}
I=\sum_{x\in\Phi}h_I^{(x)} g(\Vert x \Vert)
\end{equation}
where $\Phi$ denotes the set of all interfering nodes which is modeled as a point process on $\mathbb{R}^d$ and $h_I^{(x)}$ is a positive random variable capturing the (power) fading coefficient between the receiver and the $x^{\text{th}}$ interfering node. Here, typically $d=2$ or $d=3$, though this assumption is not necessary. The spatial region containing the interferers is assumed to be an infinite area \cite{Baccelli2009a, Haenggi2009b, Haenggi2008, Win2009}. Moreover, $\{h_I^{(x)}\}_x$ are i.i.d. random variables and independent of the point process. The path-loss is captured by a function $g(\cdot):\mathbb{R}^{+}\rightarrow\mathbb{R}^{+}$ which is a continuous, positive, non-increasing function of $\Vert x \Vert$ and assumed to depend only on the Euclidean distance $\Vert x \Vert$ from the node $x$ to the receiver at the origin. We consider the following general path-loss model \cite{Ilow1998, Haenggi2008, Win2009, Baccelli2009a}:
\begin{equation}
\label{eqn:path_loss_model}
g(\Vert x \Vert)=(a+b \Vert x \Vert^{\delta})^{-1}
\end{equation}
for some $b > 0$, $\delta > d$ and $a \in \lbrace 0,1 \rbrace$, where $\delta$ is called the path-loss exponent, $a$ determines whether the path-loss model belongs to a singular path-loss model ($a=0$) or a non-singular path-loss model ($a=1$), and $b$ is a compensation parameter to keep the total receive power normalized which will be discussed in detail in Section \ref{sec:Comparison of Path-loss Models}. It is noted that from Campbell's theorem in Section \ref{sec:Campbell_Theorem} and the aggregated interference model in \eqref{eqn:Interference_model}, any two stationary point processes with the same intensity have equal mean power of interferences (when the expectation of \eqref{eqn:Interference_model} exists).

Assuming the effective channel power between the desired receiver and its transmitter is $h_S$ which can be a non-negative constant or a RV, the signal to interference ratio (SIR) is given by
\begin{equation}
\label{eqn:SIR_def}
\emph{SIR} = \frac{h_S}{I}
\end{equation}
where $I$ is the interference power given by \eqref{eqn:Interference_model}. If there is additive noise with power $W$ which can be a non-negative constant or RV as well, then the signal to interference plus noise ratio (SINR) is given by \cite{Haenggi2008, Haenggi2009a, Baccelli2009a}
\begin{equation}
\label{eqn:SINR_def}
\emph{SINR} = \frac{h_S}{W+I} \text{.}
\end{equation}
In this paper, the LT ordering of two interference distributions will be mainly discussed. The Laplace transform of interference plus noise is $\mathbb{E}[\exp{(-s(W+I))}]=\mathcal{L}_{W}(s)\mathcal{L}_{I}(s)$. Clearly the LT ordering of two interference distributions is not affected by a common noise power $W$. Therefore, we focus on interference distributions hereafter.

\section{Ordering of Outage Probability and Ergodic Capacity Metric}
\label{sec:Ordering_of_Perform_Metrics}
In this section we introduce performance metrics involving the stochastic ordering of interference. Firstly, we study the SIR-based outage probability. It has been shown in \cite{Haenggi2008} that when $h_S$ is exponentially distributed, the LT ordering of the interference can be related to the outage defined in terms of SINR. We generalize this result to a broader class of distributions for the effective channel RV $h_S$.
\begin{theorem}
\label{suffi_condi_st_Rayleigh}
Let $I_{1}$ and $I_{2}$ denote the interferences from point process $\Phi_1$ and $\Phi_2$ respectively. Also, let $h_S$ be the effective fading channel between the desired receiver and its transmitter and has a CCDF $\bar{F}_{h_S}(x):=1-F_{h_S}(x)$, which is a c.m. function. Under these assumptions, if $I_{1} \leq_{\mathrm{Lt}} I_{2}$, then $\text{SIR}_2 \leq_{\mathrm{st}} \text{SIR}_1$.
\end{theorem}
\begin{proof}
Let $\bar{F}_{h_S}(x)$ be the c.m. CCDF of $h_S$. Then we have
\begin{eqnarray}
\lefteqn{P(\emph{\text{SIR}}_1>x) = \mathbb{E}_{I_{1}}\left[P(h_S>xI_{1})\right] = \mathbb{E}_{I_{1}}\left[\bar{F}_{h_S}(xI_{1})\right]} \\
&\geq& \mathbb{E}_{I_{2}}\left[\bar{F}_{h_S}(xI_{2})\right] = \mathbb{E}_{I_{2}}\left[P(h_S>xI_{2})\right] = P(\emph{\text{SIR}}_2>x) \text{,}
\end{eqnarray}
where the inequality is due to $I_{1} \leq_{\mathrm{Lt}} I_{2}$ and equation \eqref{eqn:def_Lt_ordering_conseq1}. Recalling the definition of $\leq_{\mathrm{st}}$ in \eqref{eqn:def_st_ordering}, Theorem \ref{suffi_condi_st_Rayleigh} is proved.
\end{proof}
Theorem \ref{suffi_condi_st_Rayleigh} enables us to conclude the usual stochastic ordering of SIRs whenever LT ordering of interferences are established. In the sequel, we will show many examples of how LT ordering between two interference distributions can be established which can be used with Theorem \ref{suffi_condi_st_Rayleigh} to establish the ``usual stochastic" ordering of outage as in \eqref{eqn:def_st_ordering}. Illustrations of Theorem \ref{suffi_condi_st_Rayleigh} are shown in our numerical results in Section \ref{sec:Numerical_Results}. Note that if two interferences have the usual stochastic ordering $I_{1} \leq_{\mathrm{st}} I_{2}$ (a stronger assumption than in Theorem \ref{suffi_condi_st_Rayleigh}), then we can have $\emph{\text{SIR}}_2 \leq_{\mathrm{st}} \emph{\text{SIR}}_1$ with any arbitrary distribution for $h_S$ and not just those with a c.m. CCDF.

Now consider the expression
\begin{equation}
\label{eqn:Capacity_def}
C=\mathbb{E}\left[\log_2\left(1+\frac{h_S}{W+I}\right)\right] \text{,}
\end{equation}
which is the bandwidth-normalized capacity in the weak interference regime \cite{Haenggi2008, Shang2009, Haenggi2009a}. The expectation in \eqref{eqn:Capacity_def} is with respect to the arbitrary positive random variable $h_S$, as well as the interference $I$ given in \eqref{eqn:Interference_model}. The noise power $W \geq 0$ is fixed, and $h_S, I$ are independent.

We would like to compare two regimes ($h_{S_1}, I_1$) with capacity $C_1$, and ($h_{S_2}, I_2$) with capacity $C_2$, where $h_{S_2} \leq_{\mathrm{Lt}} h_{S_1}$, and $I_{1} \leq_{\mathrm{Lt}} I_{2}$. Since $\log(1+h_S)$ is c.m.d. with respect to $h_S$ and $\log(1+1/(W+I))$ is c.m. with respect to $I$, by using \eqref{eqn:def_Lt_ordering_conseq1} and \eqref{eqn:def_Lt_ordering_conseq2} successively, we have $C_1 \geq C_2$. Therefore the LT ordering of interferences and the LT ordering of effective channels together lead to ordering of ergodic capacities. Unlike in Theorem \ref{suffi_condi_st_Rayleigh}, in this case, we can compare the ergodic capacity \emph{regardless of the distribution} of fading channel, $h_S$. Section \ref{Parametric Fading on the Interference Link} has examples of fading distributions that are LT ordered.

Having emphasized the impact of LT ordering on outage and capacity metrics, in what follows, we will investigate the conditions for the presence of LT ordering of interference distributions. The three factors which affect the interference distribution are the fading channel from the interfering nodes to the receiver $h_I^{(x)}$, the path-loss model $g(\cdot)$, and the point process of interfering nodes $\Phi$ as shown in the interference model in \eqref{eqn:Interference_model}. We derive the conditions on fading channels and path-loss models and the underlying point processes for the LT ordering of interferences. We also identify the LT ordering of interference distributions in commonly used point processes.

\section{Comparison of Fading Channels on the Interference Link}
\label{sec:Comparison of Fading_Channels}
In the previous section we mentioned that the effective channels between the transmitter and receiver, $h_{S_1}$ and $h_{S_2}$, can be compared. In this section, we find the conditions on the distribution of $h_I^{(x)}$ on the \emph{interference link} for the interference in \eqref{eqn:Interference_model} to be LT ordered. Since $h_I^{(x)}$ are assumed as i.i.d., we will drop the node index $x$ for convenience hereafter, when we refer to its distribution.
\begin{theorem}
\label{Lt_ordering_w_fading_for_GPP}
Let $I_1$ and $I_2$ denote the interferences both with the same path-loss model $g(\Vert x \Vert)$, from a stationary point process $\Phi$ as in \eqref{eqn:Interference_model}. Also, let $h_{I_1}$ and $h_{I_2}$ be RVs whose distributions capture the fading channels between the receiver and interferers under the two scenarios that are compared. Under these assumptions, if $h_{I_1} \leq_{\mathrm{Lt}} h_{I_2}$, then $I_1 \leq_{\mathrm{Lt}} I_2$.
\end{theorem}
\begin{proof}
The proof is shown in Appendix A.
\end{proof}
The interference in a stationary point process depends on the fading channel between the receiver and interferers. Intuitively a bigger LT ordering indicates a fading channel that is more like AWGN (i.e., ``less fading''). Therefore, the interference arising from it indicates interference adding up more coherently, giving rise to a ``bigger" interference in the LT ordering sense.

As mentioned in Section \ref{subsec:System_Model}, we assumed the interferers are distributed on an infinite region due to the stationarity of the point process. However, receivers in many wireless networks may experience interference from finite-area regions \cite{Salbaroli2009}. Theorem \ref{Lt_ordering_w_fading_for_GPP} is also valid for the interference in finite area. Moreover, when guard zones around the receiver occur due to sophisticated MAC protocols \cite{Hasan2007, Baccelli2009b}, Theorem 2 still holds.

\subsection{Parametric Fading on the Interference Link}
\label{Parametric Fading on the Interference Link}
In this section, we show the interference distributions are monotonic in line of sight (LoS) parameter of the fading channels with respect to the LT order for commonly used parametric fading distributions such as Nakagami-$m$ and Ricean fading. By Theorem \ref{Lt_ordering_w_fading_for_GPP} this implies LT ordering of interferences, which, by Theorem \ref{suffi_condi_st_Rayleigh} establishes stochastic ordering of SIRs.

Consider Nakagami fading model, where the envelope $\sqrt{h_I}$ is Nakagami and the effective channel $h_I$ follows the distribution
\begin{equation}
\label{eqn:Gamma_dist}
f_{h_I}(x)=\frac{m^m}{\Gamma(m)}x^{m-1}\exp(-mx),\text{ } x \geq 0 \text{.}
\end{equation}
Since $\mathbb{E}[-sh_I]=(1+s/m)^{-m}$ is a decreasing function of $m$ for each $s$, it follows that if the $m$ parameters of two channel distributions satisfy $m_1 \leq m_2$ then, $h_{I_1} \leq_{\mathrm{Lt}} h_{I_2}$ where $h_{I_1}$ and $h_{I_2}$ have normalized Gamma distributions with parameter $m_1$ and $m_2$ respectively. From Theorem \ref{Lt_ordering_w_fading_for_GPP}, it follows that $I_1 \leq_{\mathrm{Lt}} I_2$, if $m_1 \leq m_2$.

Similarly, the envelope $\sqrt{h_I}$ is Ricean and the distribution of effective channel $h_I$ is given by
\begin{equation}
\label{eqn:Ricean_effec_chan}
f_{h_I}(x)=(K+1)\exp\left[-(K+1)x-K\right]I_0(2\sqrt{K(K+1)x}),\text{ } x \geq 0 \text{,}
\end{equation}
where $K$ is the LoS parameter of Ricean fading channel and $I_0(t):=\sum_{m=0}^{\infty}(t/2)^{2m}/(m!\Gamma(m+1))$ is the modified Bessel function of the first kind of order zero. The Laplace transform of \eqref{eqn:Ricean_effec_chan} decreases with $K$ for all $s \geq 0$. Thus, similar to the Nakagami case, if $K_1 \leq K_2$ are the Ricean parameters of two channels, then $h_{I_1} \leq_{\mathrm{Lt}} h_{I_2}$ which by Theorem \ref{Lt_ordering_w_fading_for_GPP} imply that $I_1 \leq_{\mathrm{Lt}} I_2$.

In addition to the Nakagami-$m$ and Ricean fading cases, it can be shown through a procedure similar to the discussion that the interference distribution corresponding to Nakagami-$q$ (Hoyt) fading \cite{Simon} also satisfies LT ordering with respect to the shape parameter.

\subsection{Interference in Combined Multipath Fading and Shadowing}
The effect of shadow fading on the interference power distribution can be modeled as a product of a shadowing random variable with a multipath fading random variable. Let $h_{I_1} \leq_{\mathrm{Lt}} h_{I_2}$ be the two effective multipath fading distributions, and $X_1 \leq_{\mathrm{Lt}} X_2$ be the two shadowing distributions. Then, from \cite[Theorem 5.A.7 (d)]{Shaked}, it follows that the composite RV satisfies $h_{I_1}X_1 \leq_{\mathrm{Lt}} h_{I_2}X_2$, since $l(x,y) = xy$ has a c.m. derivative in each variable. We conclude that if $h_{I_1} \leq_{\mathrm{Lt}} h_{I_2}$ and $X_1 \leq_{\mathrm{Lt}} X_2$, then $I_1 \leq_{\mathrm{Lt}} I_2$ from Theorem \ref{Lt_ordering_w_fading_for_GPP}. Such conclusions are especially useful even when the composite distribution of $h_{I_1}X_1$ or $h_{I_2}X_2$ cannot be written in closed-form.

\section{Comparison of Path-loss Models}
\label{sec:Comparison of Path-loss Models}
Here, we show the ordering of mean power of interferences in a stationary point process with different path-loss models. Generally, the mean power of interference is an important factor to determine a system performance. For example, by considering the mean power of interferences as noise power, the average error rate can be approximated. Thus, we are interested in comparison the mean power of interferences. In what follows, we show the condition on the mean power of interferences with non-singular path-loss models with different path-loss exponents under which this ordering holds:
\begin{theorem}
\label{mean_ordering_with_non_singular_path_loss_for_GPP}
Let $I_1$ and $I_2$ denote the interferences both with identical fading distribution $h_I$ in a stationary point process $\Phi$. Also, let $g_1(\Vert x \Vert)$ and $g_2(\Vert x \Vert)$ be the non-singular path-loss models with $a=1$ and $b=1$ in \eqref{eqn:path_loss_model} and with different path-loss exponents $\delta_1$ and $\delta_2$, respectively. If $\delta_1 \leq \delta_2$, then $\mathbb{E}[I_{2}] \leq \mathbb{E}[I_{1}]$.
\end{theorem}
\begin{proof}
The proof is given in Appendix B.
\end{proof}
Theorem \ref{mean_ordering_with_non_singular_path_loss_for_GPP} compares the mean power of interference for two different path-loss models and identical fading. It is clear from the proof of Theorem \ref{mean_ordering_with_non_singular_path_loss_for_GPP} that if the fading is different with $\mathbb{E}[h_{I_{1}}] \geq \mathbb{E}[h_{I_{2}}]$, the conclusion would still hold.

Under a stationary Poisson point process, we can establish the stronger stochastic ordering, LT ordering of interferences, because if $I_{1} \geq_{\mathrm{Lt}} I_{2}$, then $\mathbb{E}[I_{1}] \geq \mathbb{E}[I_{2}]$. In what follows, we find the conditions on the interference distributions with non-singular path-loss models with different path-loss exponents under which this stochastic ordering holds:
\begin{theorem}
\label{Lt_ordering_with_non_singular_path_loss_for_PPP}
Let $I_1$ and $I_2$ denote the interferences both with identical fading distribution $h_I$ in a stationary Poisson point process $\Phi_{\normalfont{\text{PPP}}}$. Also let $g_1(\Vert x \Vert)$ and $g_2(\Vert x \Vert)$ be the non-singular path-loss models with $a=1$ and $b=1$ in \eqref{eqn:path_loss_model} and with different path-loss exponents $\delta_1$ and $\delta_2$, respectively. If $\delta_1 \leq \delta_2$, then $I_2 \leq_{\mathrm{Lt}} I_1$.
\end{theorem}
\begin{proof}
The proof is given in Appendix C.
\end{proof}

For the non-singular path-loss model $g(\Vert x \Vert)$ with $a=1$ and $b=1$ is assumed, however, the mean power of interference with $g_1(\Vert x \Vert)$ is greater than that with $g_2(\Vert x \Vert)$ from Theorem \ref{mean_ordering_with_non_singular_path_loss_for_GPP}. Thus, for more fair comparison of LT ordering of the interferences, we set the parameter $b$ in \eqref{eqn:path_loss_model} for $g_2(\Vert x \Vert)$ to have a equal mean power of interference as that of $g_1(\Vert x \Vert)$, using Campbell's Theorem as follows:
\begin{equation}
\label{eqn:parameter_b}
b=\left(\frac{\delta_2\Gamma\left(1-\frac{d}{\delta_1}\right)\Gamma\left(\frac{d}{\delta_1}\right)}{\delta_1\Gamma\left(1-\frac{d}{\delta_2}\right)\Gamma\left(\frac{d}{\delta_2}\right)}\right)^{-\frac{\delta_2}{d}} \text{,}
\end{equation}
where $b \leq 1$. It is noted that the interferences, $I_1$ and $I_2$, with path-loss models, $g_1(\Vert x \Vert)$ with $a=1$ and $b=1$ and $g_2(\Vert x \Vert)$ with $a=1$ and $b$ is set the value in \eqref{eqn:parameter_b} and with different path-loss exponents $\delta_1$ and $\delta_2$ in \eqref{eqn:path_loss_model} in a stationary point process $\Phi$ have same mean power of interference, even if $\delta_1 \leq \delta_2$. Then, with \eqref{eqn:parameter_b}, we can establish additional LT ordering of interferences as follows:
\addtocounter{corollary}{4}
\begin{corollary}
\label{Lt_ordering_for PPP_theorem}
Let $I_1$ and $I_2$ denote the interferences both with identical fading distribution $h_I$ in a stationary Poisson point process $\Phi_{\normalfont{\text{PPP}}}$. Also let $g_1(\Vert x \Vert)$ be the non-singular path-loss model in \eqref{eqn:path_loss_model} with $\delta_1$, $a=1$ and $b=1$ and $g_2(\Vert x \Vert)$ be the non-singular path-loss models with $\delta_2$, $a=1$ and $b$ is set the value in \eqref{eqn:parameter_b}. If $\delta_1 \leq \delta_2$, then $I_2 \leq_{\mathrm{Lt}} I_1$.
\end{corollary}
\begin{proof}
The proof is the same as the proof for Corollary \ref{Lt_ordering_with_non_singular_path_loss_for_PPP} with setting the value of $b$ for $g_2(\Vert x \Vert)$ is \eqref{eqn:parameter_b} instead of $b=1$.
\end{proof}
From Corollary \ref{Lt_ordering_for PPP_theorem}, it is seen that the interference distributions with non-singular path-loss models in stationary Poisson point process are monotonic in the path-loss exponent with respect to the LT order. Indeed, the parameter $b$ in \eqref{eqn:parameter_b} was set to ensure the mean interference powers are equal. Hence, the better performance in SIR-based outage probability or ergodic capacity with increased path-loss exponent $\delta$ is not due to an improvement of average interference power.

Theorem \ref{mean_ordering_with_non_singular_path_loss_for_GPP}, \ref{Lt_ordering_with_non_singular_path_loss_for_PPP} and Corollary \ref{Lt_ordering_for PPP_theorem} show that one path-loss model need not dominate the other pointwise to get an ordering in interference in the case of the non-singular path-loss model. Indeed, if $g_1(r) \geq g_2(r), \forall r\in\mathbb{R}^+$ (for example, a comparison of singular and non-singular path-loss model), it is obvious that $\mathbb{E}[I_{1}] \geq \mathbb{E}[I_{2}]$ since the aggregated interference power with $g_1(r)$ is always greater than that with $g_2(r)$. For a stationary Poisson point process, we can also easily observe when $g_1(r) \geq g_2(r)$, $I_1 \geq_{\mathrm{Lt}} I_2$ from \eqref{eqn:proof_suffi_condi_path_loss_PPP_1}. However, from Theorem \ref{mean_ordering_with_non_singular_path_loss_for_GPP}, \ref{Lt_ordering_with_non_singular_path_loss_for_PPP} and Corollary \ref{Lt_ordering_for PPP_theorem}, we can establish the stochastic orderings of interferences even in the case that $g_1(r) < g_2(r)$ in some range of $r$ and $g_1(r) > g_2(r)$ in another range of $r$ as shown in Fig. \ref{fig:non_sin_path_loss_models}. The simulation result to verify these theorems will be shown in Section \ref{sec:Numerical_Results}.

For the singular path-loss model ($a=0$ and $b>0$) in \eqref{eqn:path_loss_model}, a finite mean of interference power does not exist in a stationary point process since \eqref{eqn:Proof_for_non_sin_path_loss_GPP_1} does not converge. Thus the means of interferences cannot be compared. In case of a stationary Poisson point process, we have a closed-form expression for Laplace transform of interference with the singular path-loss model in a stationary Poisson point process as follows \cite{Haenggi2009a}:
\begin{equation}
\label{eqn:Lt_PPP}
\mathcal{L}_{I}(s) = \exp(-\lambda c_d \mathbb{E}[h_{I}^{\alpha}]\Gamma(1-\alpha)s^{\alpha}),
\end{equation}
where $\alpha=d/\delta$ and $c_d$ is the volume of the $d$-dimensional unit ball. $\mathbb{E}[h_{I}^{\alpha}]$ is a fractional moment of $h_{I}$ with $0<\alpha<1$. Unlike the non-singular path-loss model, however, the LT ordering does not hold between two different interferences corresponding to path-loss exponents $\delta_1 \leq \delta_2$ in case of the singular path-loss model since \eqref{eqn:Lt_PPP} is not a decreasing function of $\delta$.

\section{Comparison of Different Point Processes}
In \cite{Blaszczyszyn2009} directionally convex ordering (dcx) is used to order point processes. In this section, we define a new stochastic ordering between point processes and state some results involving the ordering of point processes. We first define a new stochastic ordering of point processes based on the well-known Laplace functional:
\begin{definition}
\label{Def_Lf_ordering}
Let $\Phi_1$ and $\Phi_2$ be two stationary point processes such that
\begin{equation}
\label{eqn:Def_Lf_ordering}
L_{\Phi_1}(u)=\mathbb{E}\left[e^{-\sum_{x\in\Phi_1}u(x)}\right] \geq \mathbb{E}\left[e^{-\sum_{x\in\Phi_2}u(x)}\right]=L_{\Phi_2}(u)
\end{equation}
where $u(\cdot)$ runs over the set $\mathscr{U}$ of all non-negative functions on $\mathbb{R}^d$. Then $\Phi_1$ is said to be smaller than $\Phi_2$ in the Laplace functional (LF) order (denoted by $\Phi_1 \leq_{\mathrm{Lf}} \Phi_2$).
\end{definition}
It can be shown that LF ordering follows from dcx ordering, which makes LF ordering easier to verify and easier to relate to interference metrics. Note that the LT ordering in \eqref{eqn:def_Lt_ordering} is for RVs, whereas the LF ordering in \eqref{eqn:Def_Lf_ordering} is for point processes. They can be connected in the following way:
\begin{equation}
\label{eqn:Lf_to Lt_ordering}
\Phi_1 \leq_{\mathrm{Lf}} \Phi_2 \Longleftrightarrow \sum_{x\in\Phi_1}u(x) \leq_{\mathrm{Lt}} \sum_{x\in\Phi_2}u(x) \text{,  } \forall u \in \mathscr{U} \end{equation}
Hence, it is possible to think of LF ordering of point processes as the LT ordering of their interferences in the absence of fading, for all non-negative path-loss functions. But as we will see, LT ordering of their interferences in the presence of fading can be proved when two point processes are LF ordered. We next prove a generalization of Theorem \ref{Lt_ordering_w_fading_for_GPP} where the two point processes are different and LF ordered.
\begin{theorem}
\label{Lt_ordering_interference_from_PP}
Let $\Phi_1$ and $\Phi_2$ be two stationary point processes and $h_{I_1}$ and $h_{I_2}$ be RVs whose distributions capture the fading channels between the receiver and interferers under the two scenarios that are compared. Also let $I_1$ and $I_2$ denote the interferences with same path-loss model $g(\Vert x \Vert)$ in $\Phi_1$ and $\Phi_2$ respectively. If $\Phi_1 \leq_{\mathrm{Lf}} \Phi_2$ and $h_{I_1} \leq_{\mathrm{Lt}} h_{I_2}$, then $I_1 \leq_{\mathrm{Lt}} I_2$.
\end{theorem}
\begin{proof}
Theorem \ref{Lt_ordering_interference_from_PP} follows from Theorem \ref{Lt_ordering_w_fading_for_GPP} and equation \eqref{eqn:Lf_to Lt_ordering} with $u(x)=h_Ig(\Vert x \Vert)$.
\end{proof}
Note that in Theorem \ref{Lt_ordering_interference_from_PP}, $h_I$ does not only capture a fading distribution, but also can capture a random thinning property of point process.

\subsection{Laplace Functional Ordering between Specific Point Processes}
Here, we compare the LF order between different point processes which are commonly used. These are Poisson point process and Neyman-Scott process with Poisson distributed number of points in each cluster which is one of example of Poisson cluster process. The Neyman-Scott process results from homogeneous independent clustering applied to a stationary Poisson point process. The details for these point processes can be found in \cite{Stoyan1995, Haenggi2009b}.
\begin{theorem}
\label{Lf_PPP_vs_PCP}
If $\Phi_{\normalfont{\text{PPP}}}$ and $\Phi_{\normalfont{\text{PCP}}}$ denote Poisson point process and Neyman-Scott process with Poisson distributed number of daughter points respectively and both point processes have same intensity $\lambda$, then $\Phi_{\normalfont{\text{PCP}}} \leq_{\mathrm{Lf}} \Phi_{\normalfont{\text{PPP}}}$.
\end{theorem}
\begin{proof}
The proof is given in Appendix D.
\end{proof}
We consider another point process, the mixed Poisson process which is a simple instance of a Cox process. It can be thought of as a stationary Poisson point process with randomized intensity parameter $X$ which has the averaged intensity measure $\mathbb{E}_{X}[X]=\lambda$. Every sample of such a process looks like a sample of some stationary Poisson point process. We compare the LF ordering of this point process with that of stationary Poisson point process as follows:
\begin{theorem}
\label{Lf_PPP_vs_MPP}
If $\Phi_{\normalfont{\text{PPP}}}$ and $\Phi_{\normalfont{\text{MPP}}}$ denote Poisson point process and mixed Poisson process respectively and both point processes have same average intensity $\lambda$, then $\Phi_{\normalfont{\text{MPP}}} \leq_{\mathrm{Lf}} \Phi_{\normalfont{\text{PPP}}}$.
\end{theorem}
\begin{proof}
The proof is given in Appendix E.
\end{proof}
The aggregated interference from Poisson cluster process or mixed Poisson process is always less than that from Poisson point process in LT order by Theorem \ref{Lt_ordering_interference_from_PP}. This means that the orderings of SIR-based outage probabilities or ergodic capacities are established depending on a point process by the performance metrics in Section \ref{sec:Ordering_of_Perform_Metrics}. It is noted that from Campbell's theorem in Section \ref{sec:Campbell_Theorem}, any two stationary point processes with same intensity have equal mean power of interferences (when the expectation in \eqref{eqn:Interference_model} exists). Hence, the better performance in SIR-based outage probability or ergodic capacity in specific point process is not due to an improvement of average interference power.

The above point processes are less than Poisson point process in LF order. In what follows, we show the point process which has larger LF ordering than Poisson point process even though it is non-stationary and non-isotropic. In binomial point processes, there are a total of $N$ transmitting nodes uniformly distributed in a $d$-dimensional ball of radius $r$ centered at the origin, denoted as $B_0(r)$. The density of the process is given by $\lambda=N/(c_d r^d)$ where $c_d$ is the volume of the $d$-dimensional unit ball \cite{Srinivasa2007}. In this case, we can compare the LF ordering of point processes in bounded area as follows:
\begin{theorem}
\label{Lf_PPP_vs_BPP}
Let $\Phi_{\normalfont{\text{PPP}}}(r)$ be a Poisson point process over $B_0(r)$. If $\Phi_{\normalfont{\text{PPP}}}(r)$ and $\Phi_{\normalfont{\text{BPP}}}$ denote Poisson point process and binomial point process respectively in finite area and both point processes have same intensity $\lambda$, then $\Phi_{\normalfont{\text{PPP}}}(r) \leq_{\mathrm{Lf}} \Phi_{\normalfont{\text{BPP}}}$ whenever $0 \leq \lambda\int_{B_0(r)}\left[1-\exp(-u(x))\right]\mathrm{d} x \leq N$ holds.
\end{theorem}
\begin{proof}
The proof is given in Appendix F.
\end{proof}
In a finite area $B_0(r)$, the aggregated interference at the origin from a binomial point process is always larger than that from Poisson point process in LT order from Theorem \ref{Lt_ordering_interference_from_PP}.

\subsection{Laplace Transform Ordering of Interferences in Heterogeneous Networks}
As unlicensed band utilization increases, the unlicensed wireless network may experience adverse interference from collocated wireless devices that are transmitting in the same unlicensed band. Such a heterogeneous network scenario can be modeled as a superposition of mutually independent point processes \cite{Heath2012}. Let $\Phi_1=\bigcup_{i=1}^M \Phi_{1,i}$ and $\Phi_2=\bigcup_{i=1}^M \Phi_{2,i}$ for $i=1,...,M$ be the heterogeneous networks which are modeled as superpositions of mutually independent point processes. Since the Laplace functional of superposition of mutually independent point processes is $L_{\Phi}(u)=\prod_{i=1}^{M}L_{\Phi_i}(u)$ \cite{Daley2008}, if $\Phi_{1,i} \leq_{\mathrm{Lf}} \Phi_{2,i}$ for $i=1,...,M$, then $\Phi_1 \leq_{\mathrm{Lf}} \Phi_2$. Therefore, $I_{1} \leq_{\mathrm{Lt}} I_{2}$ from Theorem \ref{Lt_ordering_interference_from_PP}.

\section{Numerical Results}
\label{sec:Numerical_Results}
In this section, we verify our theoretical results through Monte Carlo simulations. Since the LT ordering between two interference scenarios cannot be verified directly from its probability distributions such as PDF and CDF, we will verify the LT ordering between interferences by the ordering of SIR-based outage probabilities or ergodic capacities as mentioned in Section \ref{sec:Ordering_of_Perform_Metrics}.

\subsection{Comparison of Fading Channels on the Interference Link}
In many practical scenarios, different links in wireless networks can experience asymmetric fading conditions. If the interferer's channel is Nakagami-$m$ fading, while the desired link is Rayleigh fading which has an effective channel CCDF which is c.m., we can compare SIR-based outage probabilities using using Theorem \ref{suffi_condi_st_Rayleigh}.

In Fig. \ref{fig:different_m_PCP} the CDFs of interference power and SIR from Poisson cluster process with different Nakagami-$m$ fading parameters and with the non-singular path-loss model ($a=1$, $b=1$ and $\delta=4$) in \eqref{eqn:path_loss_model} are shown. We consider two different LoS parameters: $m_1=1$ and $m_2=2$. The choice of these parameters ensures $I_1 \leq_{\mathrm{Lt}} I_2$ from Theorem \ref{Lt_ordering_w_fading_for_GPP}. Consequently, we observe $\emph{SIR}_1 \geq_{\mathrm{st}} \emph{SIR}_2$ in the bottom of Fig. \ref{fig:different_m_PCP} which agrees with Theorem \ref{suffi_condi_st_Rayleigh} even though there is a crossover point between interference power distributions in the top of Fig. \ref{fig:different_m_PCP}.

The CDFs of SIR from Poisson point process with different Nakagami-$m$ fading parameters and with the singular path-loss model ($a=0$, $b=1$ and $\delta=4$) in \eqref{eqn:path_loss_model} are shown in Fig. \ref{fig:different_m_PPP}. Similarly the LoS parameters $m_1 \leq m_2$ lead to $I_1 \leq_{\mathrm{Lt}} I_2$ from Theorem \ref{Lt_ordering_w_fading_for_GPP}. Clearly, it is observed $\emph{SIR}_1 \geq_{\mathrm{st}} \emph{SIR}_2$ in Fig. \ref{fig:different_m_PPP} which agrees with Theorem \ref{suffi_condi_st_Rayleigh}.

Table \ref{table:Ergodic_capacity_PCP} shows the ergodic capacity performances when the desired link has Ricean fading channel with $K_S=5$ and interfering channels follow Ricean distributions with $K_{I_1}=0$ and $K_{I_2}=1$ in the two scenarios compared on a Poisson cluster process. The ergodic capacity with $I_1$ is always better than that with $I_2$ as expected since the interference distributions are monotonic in LoS parameter of Ricean fading channel with respect to the LT ordering in Section \ref{sec:Comparison of Fading_Channels}.

The ergodic capacities in Poisson point process are shown in Table \ref{table:Ergodic_capacity_PPP}. In this case, all conditions are same except for the type of point process. Therefore, the ergodic capacity with $I_1$ is always better than that with $I_2$ when $I_1 \leq_{\mathrm{Lt}} I_2$.

\subsection{Comparison of Path-loss Models}
We show in Fig. \ref{fig:LT_ordering_non_sin_path_loss_models} the CDFs of the interference power and CDFs of SIR from a Poisson point process with the non-singular path-loss models, $g_1(\Vert x \Vert)$ and $g_2(\Vert x \Vert)$ which are given in Fig. \ref{fig:non_sin_path_loss_models} and discussed in Corollary \ref{Lt_ordering_for PPP_theorem}. It is noted the non-singular path-loss models with two different path-loss exponents: $\delta_1=4$ and $\delta_2=8$ ensures $I_1 \geq_{\mathrm{Lt}} I_2$ from Corollary \ref{Lt_ordering_for PPP_theorem}. We consider additional non-singular path-loss, $g_3(\Vert x \Vert)$ whose parameters are $a=1, b=1$, and $\delta=8$ in Fig. \ref{fig:non_sin_path_loss_models}. Since $g_3(\Vert x \Vert) \leq g_2(\Vert x \Vert)$ for $\Vert x \Vert \geq 0$ as shown in Fig. \ref{fig:non_sin_path_loss_models}, it is obvious $I_2 \geq_{\mathrm{Lt}} I_3$. Thus, we observe $\emph{SIR}_1 \leq_{\mathrm{st}} \emph{SIR}_2 \leq_{\mathrm{st}} \emph{SIR}_3$  in the bottom of Fig \ref{fig:LT_ordering_non_sin_path_loss_models}. By Theorem \ref{suffi_condi_st_Rayleigh}, when the fading channel between the desired receiver and its transmitter is Rayleigh distributed and the effective fading channel is exponentially distributed, we can observe the usual stochastic ordering of SIR distributions if the interferences are LT ordered as shown in Fig. \ref{fig:LT_ordering_non_sin_path_loss_models}.

\subsection{Comparison of Different Point Processes}
In the following, we compare the CDFs of interference power and CDFs of SIR from Poisson point process and Poisson cluster process. Since the interferences from Poisson point process and Poisson cluster process with the same intensity $\lambda$, the same non-singular path-loss model $g(\Vert x \Vert)$, and identical fading distribution $h_I$ yield $I_{\text{PCP}} \leq_{\mathrm{Lt}} I_{\text{PPP}}$, which implies $\emph{SIR}_{\text{PCP}} \geq_{\mathrm{st}} \emph{SIR}_{\text{PPP}}$. This is observed in the bottom of Fig. \ref{fig:PPP_vs_PCP_non_sin} as predicted from our theoretical result in Theorem \ref{Lf_PPP_vs_PCP}.

In addition to SIR-based outage performance, it is observed that the ergodic capacity in Poisson cluster processes is always greater than that in Poisson point processes by comparing same rows in Table \ref{table:Ergodic_capacity_PCP} and \ref{table:Ergodic_capacity_PPP}.

\section{Summary}
\label{sec:Summary}
In this paper, we used stochastic orders to compare performance in wireless networks. We showed that when interference is LT ordered, it is possible to order the SIR in the usual stochastic ordering sense when the effective channel has a c.m. CCDF. Similar results hold when the metric is the bandwidth-normalized capacity. This lead to the study of the conditions for LT ordering of interference. Three factors affecting interference are the fading channel from the interfering nodes to the receiver, the path-loss model and the distribution of the interfering node location. We derived conditions on these factors so that LT ordering between interferences holds. In addition, we defined Laplace functional ordering of point processes and derived its inherent stochastic ordering of interferences when the fading channel and the path-loss model are assumed to be same for both point processes. The power of this approach is that such comparisons can be made even in cases where a closed form expression for the interference is not analytically tractable. We verified our results through Monte Carlo simulations.

\appendices
\section*{Appendix A: Proof of Theorem \ref{Lt_ordering_w_fading_for_GPP}}
\label{app:Lt_GPP}
The Laplace transform of aggregated interference $I:=\sum_{x\in\Phi}h_I^{(x)}g(\Vert x \Vert)$ is a Laplace functional \eqref{eqn:Def_Lf_of_PP} evaluated at $u(x)=s h_I g(\Vert x \Vert), s \geq 0, h_I \geq 0$ and $g(\Vert x \Vert) \geq 0$ where $h_I$ is the effective fading channel between the receiver and interferers and $g(\Vert x \Vert)$ is a path-loss model. $\mathcal{L}_{I}(s)$ can be expressed as follows:
\begin{equation}
\label{eqn:Gen_PP_1}
\mathcal{L}_{I}(s)=\mathbb{E}\left[e^{-sI}\right]=\mathbb{E}\left[e^{-sh_I\int_{\mathbb{R}^d}g(\Vert x \Vert)\Phi(\mathrm{d}x)}\right] \text{,}
\end{equation}
where the expectation is to be taken over both $\Phi$ and $h_I$. Let $Z = \int_{\mathbb{R}^d}g(\Vert x \Vert)\Phi(\mathrm{d}x)$ in \eqref{eqn:Gen_PP_1}. From \cite[Theorem 5.A.7 (b)]{Shaked}, if $\mathbb{E}\left[\exp{(-sI_1)}\vert Z=z\right] \geq \mathbb{E}\left[\exp{(-sI_2)}\vert Z=z\right]$ for all $z$ in the support of $Z$, then $I_1 \leq_{\mathrm{Lt}} I_2$. Therefore, it is sufficient to show the following equation regardless of a point process $\Phi$ in order to satisfy the LT ordering between interferences,
\begin{eqnarray}
\label{eqn:Gen_PP_2}
\mathbb{E}_{h_{I_1}}\left[e^{-szh_{I_1}}\vert Z=z\right] \geq \mathbb{E}_{h_{I_2}}\left[e^{-szh_{I_2}}\vert Z=z\right] \text{.}
\end{eqnarray}
But \eqref{eqn:Gen_PP_2} follows from the assumption $h_{I_1} \leq_{\mathrm{Lt}} h_{I_2}$. Thus, we conclude that if $h_{I_1} \leq_{\mathrm{Lt}} h_{I_2}$, then $I_1 \leq_{\mathrm{Lt}} I_2$.

\section*{Appendix B: Proof of Theorem \ref{mean_ordering_with_non_singular_path_loss_for_GPP}}
\label{app:Proof_for_non_sin_path_loss_GPP}
In order to prove $\mathbb{E}[I_{1}] \geq \mathbb{E}[I_{2}]$, using Campbell's theorem in Section \ref{sec:Campbell_Theorem} with $u(x)=h_Ig(\Vert x \Vert)$ where $h_I$ is the (power) fading coefficient and $g(\Vert x \Vert)$ is the non-singular path-loss model in \eqref{eqn:path_loss_model}, we need to show the following:
\begin{equation}
\label{eqn:Proof_for_non_sin_path_loss_GPP_1}
\lambda\int_{\mathbb{R}^d} h_I g_1(\Vert x \Vert)\mathrm{d} x \geq
\lambda\int_{\mathbb{R}^d} h_I g_2(\Vert x \Vert)\mathrm{d} x \text{.}
\end{equation}
Since $h_I$ is independent from the point process and $g(\Vert x \Vert)$ can be expressed as $g(r), r=\Vert x \Vert$ under polar coordinates, after expectation with respect to $h_I$ and change to polar coordinates, the following condition needs to be satisfied to prove Theorem \ref{mean_ordering_with_non_singular_path_loss_for_GPP}:
\begin{equation}
\label{eqn:Proof_for_non_sin_path_loss_GPP_2}
\mathbb{E}\left[h_I\right]\lambda c_d d \int_0^{\infty}g_1(r)r^{d-1} \mathrm{d} r \geq \mathbb{E}\left[h_I\right]\lambda c_d d \int_0^{\infty}g_2(r)r^{d-1} \mathrm{d} r \text{.}
\end{equation}
where $c_d$ is the volume of the $d$-dimensional unit ball and $h_I$ is the (power) fading coefficient between the receiver and interferers. Using a change of variables, we get
\begin{equation}
\label{eqn:Proof_for_non_sin_path_loss_GPP_3}
\int_{g_1(\infty)}^{g_1(0)}u\left(g_1^{-1}(u)\right)^{d-1}\frac{\partial}{\partial u}\left(g_{1}^{-1}(u)\right)\mathrm{d} u \geq \int_{g_2(\infty)}^{g_2(0)}u\left(g_2^{-1}(u)\right)^{d-1}\frac{\partial}{\partial u}\left(g_{2}^{-1}(u)\right)\mathrm{d} u \text{.}
\end{equation}
where $d$ is the dimension of point process, $g^{-1}(\cdot)$ is the inverse function of $g(\cdot)$ and $I[u \in S]=1$, if $u \in S$, and $0$ otherwise, is the indicator function. Substituting the non-singular path-loss models $g_1(r)$ and $g_2(r)$ with $a=1$ and $b=1$ in \eqref{eqn:path_loss_model} into \eqref{eqn:Proof_for_non_sin_path_loss_GPP_3}, we get
\begin{equation}
\label{eqn:Proof_for_non_sin_path_loss_GPP_4}
\int_{0}^{1}u\left( \frac{\left(\frac{1}{u}-1\right)^{\frac{d}{\delta_1}-1}}{\delta_1 u^2}-\frac{\left(\frac{1}{u}-1\right)^{\frac{d}{\delta_2}-1}}{\delta_2 u^2} \right) \mathrm{d} u
= \frac{\Gamma\left(1-\frac{d}{\delta_1}\right)\Gamma\left(\frac{d}{\delta_1}\right)}{\delta_1}
-\frac{\Gamma\left(1-\frac{d}{\delta_2}\right)\Gamma\left(\frac{d}{\delta_2}\right)}{\delta_2} \geq 0 \text{,}
\end{equation}
since $\frac{\Gamma\left(1-\frac{d}{\delta}\right)\Gamma\left(\frac{d}{\delta}\right)}{\delta}$ is a decreasing function with $\delta$ for a fixed $d$ and $\delta_1 \leq \delta_2$. The proof for Theorem \ref{mean_ordering_with_non_singular_path_loss_for_GPP} is complete.

\section*{Appendix C: Proof of Theorem \ref{Lt_ordering_with_non_singular_path_loss_for_PPP}}
\label{app:Proof_for_non_sin_path_loss_PPP}
The Laplace transform of interference power in a stationary Poisson point process with path-loss model $g_j(r), j=1,2$ in \eqref{eqn:path_loss_model} and intensity $\lambda$ can be expressed as follows:
\begin{equation}
\label{eqn:proof_suffi_condi_path_loss_PPP_1}
\mathcal{L}_{I_j}(s) = \exp\left\{-\lambda c_d d \int_0^{\infty}\left[1-\exp\left(-sh_Ig_j(r)\right)\right]r^{d-1}\mathrm{d} r\right\}
\end{equation}
where $c_d$ is the volume of the $d$-dimensional unit ball and $h_I$ is the (power) fading coefficient between the receiver and interferers. From \cite[Theorem 5.A.7 (b)]{Shaked}, it is sufficient to show $\mathcal{L}_{I_1}(s) \geq \mathcal{L}_{I_2}(s)$ in \eqref{eqn:proof_suffi_condi_path_loss_PPP_1} regardless of a distribution of $h_I$ in order to satisfy the LT ordering between interferences. To do so, the following condition needs to be satisfied after change of variables:
\begin{eqnarray}
\nonumber
\label{eqn:proof_suffi_condi_path_loss_PPP_2}
\lefteqn{\int_{g_1(\infty)}^{g_1(0)}\left[1-\exp\left(-sh_Iu\right)\right]\left(g_1^{-1}(u)\right)^{d-1}\frac{\partial}{\partial u}\left(g_{1}^{-1}(u)\right)\mathrm{d} u} \\
&\geq& \int_{g_2(\infty)}^{g_2(0)}\left[1-\exp\left(-sh_Iu\right)\right]\left(g_2^{-1}(u)\right)^{d-1}\frac{\partial}{\partial u}\left(g_{2}^{-1}(u)\right)\mathrm{d} u \text{.}
\end{eqnarray}
Substituting the same non-singular path-loss models into \eqref{eqn:proof_suffi_condi_path_loss_PPP_2}, it follows
\begin{eqnarray}
\label{eqn:proof_suffi_condi_path_loss_PPP_3}
\lefteqn{\int_{0}^{1}\left(1-\exp\left(-sh_Iu\right)\right)\left( \frac{\left(\frac{1}{u}-1\right)^{\frac{d}{\delta_1}-1}}{\delta_1 u^2}-\frac{\left(\frac{1}{u}-1\right)^{\frac{d}{\delta_2}-1}}{\delta_2 u^2} \right) \mathrm{d} u } \\
\label{eqn:proof_suffi_condi_path_loss_PPP_4}
&\geq& \left(1-\exp\left(-sh_I\right)\right)\int_{0}^{1}u\left( \frac{\left(\frac{1}{u}-1\right)^{\frac{d}{\delta_1}-1}}{\delta_1 u^2}-\frac{\left(\frac{1}{u}-1\right)^{\frac{d}{\delta_2}-1}}{\delta_2 u^2} \right) \mathrm{d} u \\
\label{eqn:proof_suffi_condi_path_loss_PPP_5}
&=& \left(1-\exp\left(-sh_I\right)\right)\underbrace{\left(\frac{\Gamma\left(1-\frac{d}{\delta_1}\right)\Gamma\left(\frac{d}{\delta_1}\right)}{\delta_1}
-\frac{\Gamma\left(1-\frac{d}{\delta_2}\right)\Gamma\left(\frac{d}{\delta_2}\right)}{\delta_2}\right)}_{A} \geq 0 \text{.}
\end{eqnarray}
\eqref{eqn:proof_suffi_condi_path_loss_PPP_4} follows from $(1-\exp(-c))u \leq 1-\exp(-cu)$ for $c \geq 0$ and $0 \leq u \leq 1$ and \eqref{eqn:proof_suffi_condi_path_loss_PPP_5} follows from $1-\exp(-sh_I) \geq 0$ for $sh_I \geq 0$ and $A \geq 0$ from \eqref{eqn:Proof_for_non_sin_path_loss_GPP_4}. Theorem \ref{Lt_ordering_with_non_singular_path_loss_for_PPP} is proved.

\section*{Appendix D: Proof of Theorem \ref{Lf_PPP_vs_PCP}}
\label{app:Lf_PPP_PCP}
Let $u(\cdot) \in \mathscr{U}$ be all non-negative functions on $\mathbb{R}^d$. The Laplace functional of Neyman-Scott process with Poisson distributed number of daughter points and with the distribution $f(\cdot)$ for locations of daughter points can be expressed as follows\cite{Stoyan1995, Haenggi2008}:
\begin{eqnarray}
\nonumber
L_{\Phi_{\text{PCP}}}(u)
\nonumber
&=&\exp\biggl\{-\lambda_p\int_{\mathbb{R}^d}\biggl[ 1-\exp\biggl(-\bar{c}\biggl(1-\int_{\mathbb{R}^d}\exp(-u(x+y))f(y)\mathrm{d} y \biggr)\biggr)\biggr]\mathrm{d} x \biggr\} \\
\label{eqn:PPP_vs_PCP_1}
&\geq&\exp\biggl\{-\lambda_p\bar{c}\int_{\mathbb{R}^d}\biggl[1-\int_{\mathbb{R}^d}\exp(-u(x+y))f(y)\mathrm{d} y\biggr]\mathrm{d} x\biggr\} \\
\label{eqn:PPP_vs_PCP_2}
&=&\exp\left\{-\lambda\int_{\mathbb{R}^d}\left[1-\exp(-u(x))\right]\mathrm{d} x\right\} \\
\nonumber
&=&L_{\Phi_{\text{PPP}}}(u)
\end{eqnarray}
where the inequality in \eqref{eqn:PPP_vs_PCP_1} follows from the fact that $1-\exp(-ax) \leq ax, \text{ } a \geq 0$ and \eqref{eqn:PPP_vs_PCP_2} follows from change of variables, interchanging integrals and using $\int f(y) \mathrm{d}y=1$ \cite{Haenggi2009b}.

\section*{Appendix E: Proof of Theorem \ref{Lf_PPP_vs_MPP}}
\label{app:Lf_PPP_MPP}
Let $u(\cdot) \in \mathscr{U}$ be all non-negative functions on $\mathbb{R}^d$. The Laplace functional of the mixed Poisson process process with a random intensity measure $X$ which has the averaged intensity measure $\mathbb{E}_{X}[X]=\lambda$ can be expressed as follows \cite{Griffiths1978}:
\begin{eqnarray}
\label{eqn:PPP_vs_MPP_1}
L_{\Phi_{\text{MPP}}}(u)&=&\mathbb{E}_{X}\left[\exp\left\{-X\int_{\mathbb{R}^d}\left[1-\exp(-u(x))\right]\mathrm{d} x\right\}\right] \\
\label{eqn:PPP_vs_MPP_2}
&\geq& \exp\left\{-\mathbb{E}_{X}[X]\int_{\mathbb{R}^d}\left[1-\exp(-u(x))\right]\mathrm{d} x\right\} \\
&=&\exp\left\{-\lambda\int_{\mathbb{R}^d}\left[1-\exp(-u(x))\right]\mathrm{d} x\right\} \\
\nonumber
&=&L_{\Phi_{\text{PPP}}}(u)
\end{eqnarray}
where the inequality in \eqref{eqn:PPP_vs_MPP_2} follows from Jensen's inequality since the term inside the brackets in \eqref{eqn:PPP_vs_MPP_1} is a convex function of $X$.
\section*{Appendix F: Proof of Theorem \ref{Lf_PPP_vs_BPP}}
\label{app:Lf_PPP_BPP}
Let $u(\cdot) \in \mathscr{U}$ be all non-negative functions on $\mathbb{R}^d$. The Laplace functional of the binomial point process consisting $N$ points with a density $\lambda$ can be expressed as follows \cite{Srinivasa2007}:
\begin{eqnarray}
\label{eqn:PPP_vs_BPP_1}
L_{\Phi_{\text{BPP}}}(u)&=&\left(1-\frac{\lambda}{N}\int_{B_0(r)}\left[1-\exp(-u(x))\right]\mathrm{d} r\right)^{N} \\
\label{eqn:PPP_vs_BPP_2}
&\leq& \exp\left\{-\lambda\int_{B_0(r)}\left[1-\exp(-u(x))\right]\mathrm{d} x\right\} \\
\nonumber
&=&L_{\Phi_{\text{PPP}}(r)}(u)
\end{eqnarray}
where the inequality in \eqref{eqn:PPP_vs_BPP_2} is due to $(1-c/n)^n \leq e^{-c}$ for $0 \leq c \leq n$. Thus, Theorem \ref{Lf_PPP_vs_BPP} is followed whenever $0 \leq \lambda\int_{B_0(r)}\left[1-\exp(-u(x))\right]\mathrm{d} x \leq N$ holds.

\bibliographystyle{IEEEtran}
\nocite{*}
\bibliography{references}
\begin{table}[ht]
\caption{\normalfont{Ergodic capacities (bits/s/Hz) over Ricean fading channel with $K_S=5$ in Poisson cluster process}}
\vspace{-0.5cm}
\centering
\renewcommand{\arraystretch}{1.0}
\scalebox{1.0}{
\begin{tabular}{c|c|c|c|c|c|c|c|c}
\hline
SINR (dB)       & -4     & -2     & 0      & 2      & 4      & 6      & 8      & 10     \\ [0.5ex]
\hline\hline
$I_1$ ($K_{I_1}=0$) & 0.9433 & 1.2282 & 1.5709 & 1.9820 & 2.4538 & 3.0063 & 3.6482 & 4.5313 \\
\hline
$I_2$ ($K_{I_2}=1$) & 0.9426 & 1.2277 & 1.5707 & 1.9803 & 2.4526 & 3.0020 & 3.6449 & 4.5278 \\
\hline
\end{tabular}}
\label{table:Ergodic_capacity_PCP}
\end{table}

\begin{table}[ht]
\caption{\normalfont{Ergodic capacities (bits/s/Hz) over Ricean fading channel with $K_S=5$ in Poisson point process}}
\vspace{-0.5cm}
\centering
\renewcommand{\arraystretch}{1.0}
\scalebox{1.0}{
\begin{tabular}{c|c|c|c|c|c|c|c|c}
\hline
SINR (dB)       & -4     & -2     & 0      & 2      & 4      & 6      & 8      & 10     \\ [0.5ex]
\hline\hline
$I_1$ ($K_{I_1}=0$) & 0.9346 & 1.2162 & 1.5485 & 1.9501 & 2.3847 & 2.9295 & 3.5491 & 4.3477 \\
\hline
$I_2$ ($K_{I_2}=1$) & 0.9342 & 1.2152 & 1.5468 & 1.9460 & 2.3816 & 2.9231 & 3.5407 & 4.3349 \\
\hline
\end{tabular}}
\label{table:Ergodic_capacity_PPP}
\end{table}

\begin{figure}[tb]
\begin{minipage}{1\textwidth}
\centering
\begin{center}
\includegraphics[height=8.5cm,keepaspectratio]{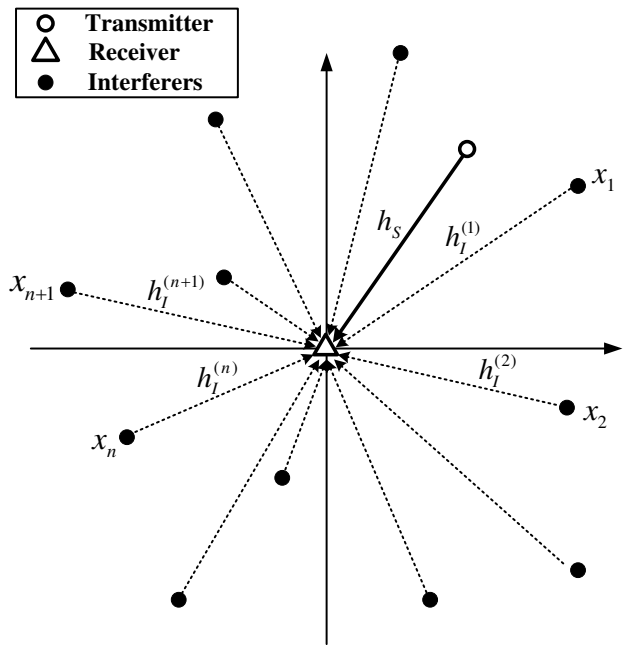}
\caption{Illustration of a wireless network. The black dots represent interfering nodes which form a point process $\Phi$ and the dotted lines represent their interfering signals. The white dot and the triangle at the origin are the desired transmit/receiver pair which are not part of the point process.}\label{fig:system_model}
\end{center}
\end{minipage}
\end{figure}

\begin{figure}[tb]
\begin{minipage}{1\textwidth}
\centering
\begin{center}
\includegraphics[height=8.5cm,keepaspectratio]{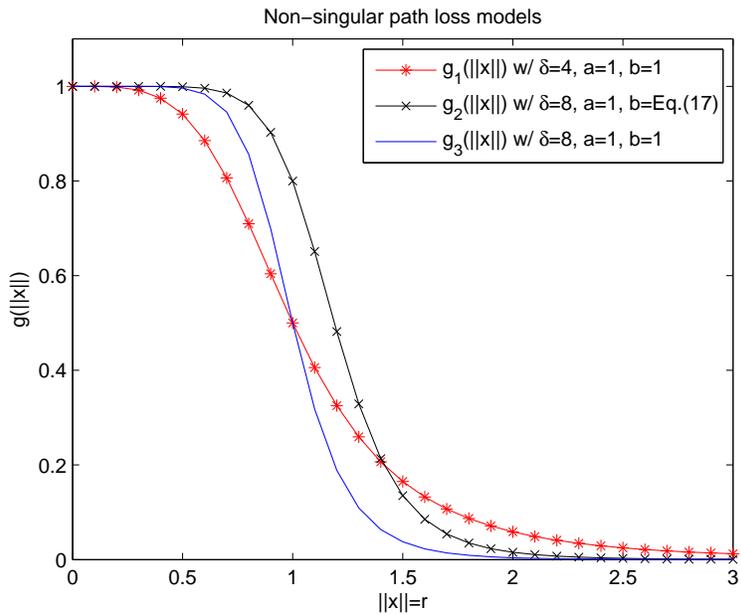}
\caption{Non-singular path-loss models with different path-loss exponents, $\delta_1$ and $\delta_2$}\label{fig:non_sin_path_loss_models}
\end{center}
\end{minipage}
\end{figure}

\begin{figure}[tb]
\begin{minipage}{1\textwidth}
\centering
\begin{center}
\includegraphics[height=8.5cm,keepaspectratio]{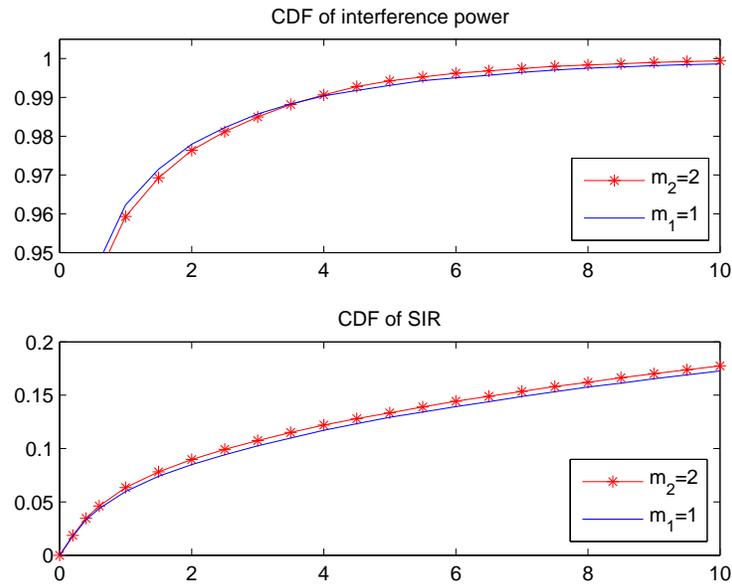}
\caption{CDFs of interference and SIR for Poisson cluster process with different fading parameters and with $\lambda=0.01$}\label{fig:different_m_PCP}
\end{center}
\end{minipage}
\end{figure}

\begin{figure}[tb]
\begin{minipage}{1\textwidth}
\centering
\begin{center}
\includegraphics[height=8.5cm,keepaspectratio]{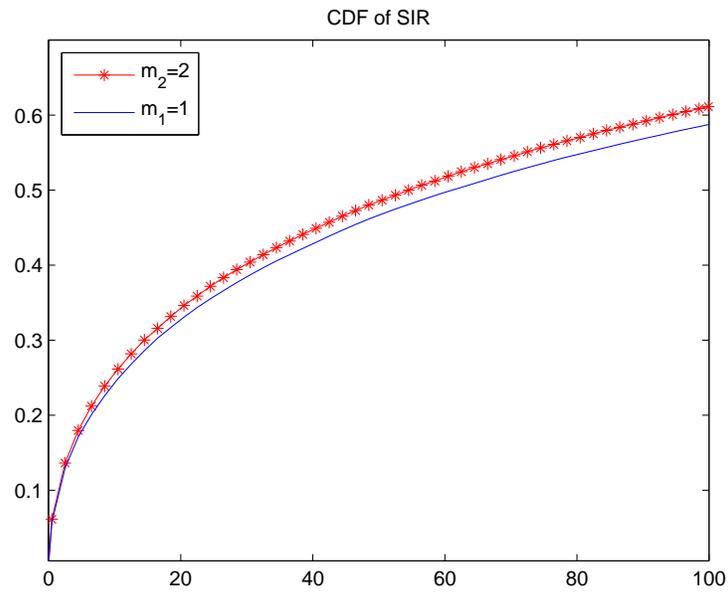}
\caption{CDFs of SIR for Poisson point process with different fading parameters and with $\lambda=0.01$}\label{fig:different_m_PPP}
\end{center}
\end{minipage}
\end{figure}

\begin{figure}[tb]
\begin{minipage}{1\textwidth}
\centering
\begin{center}
\includegraphics[height=8.5cm,keepaspectratio]{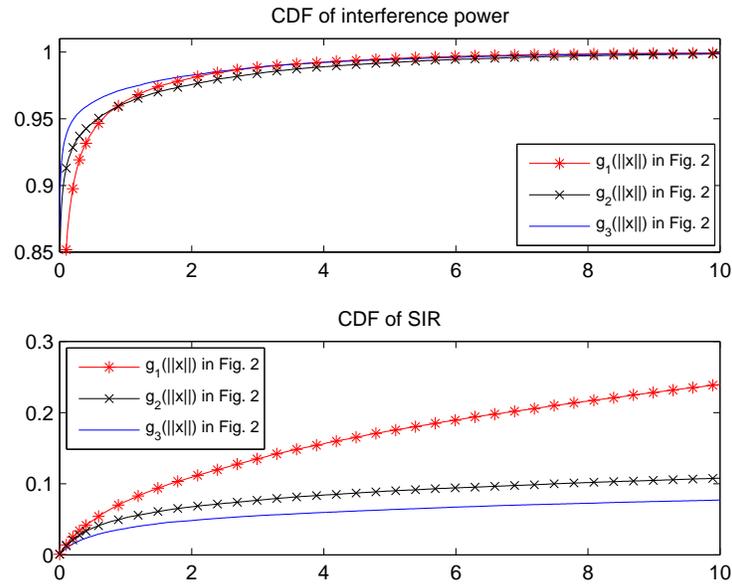}
\caption{CDFs of interference and SIR for Poisson point process with non-singular path-loss models with different path-loss exponents and with $\lambda=0.01$}\label{fig:LT_ordering_non_sin_path_loss_models}
\end{center}
\end{minipage}
\end{figure}

\begin{figure}[tb]
\begin{minipage}{1\textwidth}
\centering
\begin{center}
\includegraphics[height=8.5cm,keepaspectratio]{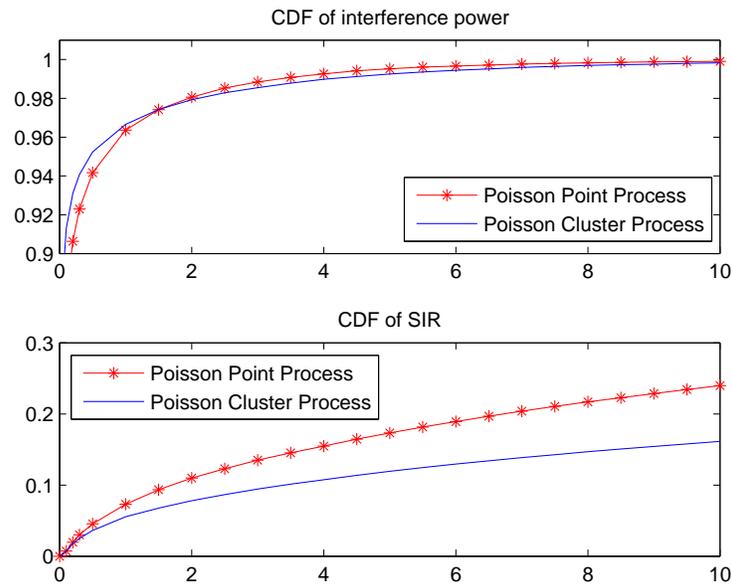}
\caption{CDFs of interference and SIR for Poisson point process and Poisson cluster process with $\lambda=0.01$}\label{fig:PPP_vs_PCP_non_sin}
\end{center}
\end{minipage}
\end{figure}

\end{document}